\documentclass[twocolumn,journal]{IEEEtran}

\usepackage{color}
\usepackage{graphicx}
\usepackage{amsmath}
\usepackage{url}
\usepackage{amssymb}
\usepackage{algorithm}
\usepackage{algorithmic}
\usepackage{amsmath}
\usepackage{multirow}
\usepackage{booktabs}
\usepackage{array}
\usepackage{amsthm}
\usepackage{caption}
\usepackage{stfloats}
\usepackage{subfigure}
\usepackage{bm}
\usepackage{booktabs}
\usepackage{setspace}
\usepackage{upgreek}
\usepackage{setspace}

\newtheorem{prop}{Proposition}

\newcommand{\be}{\begin{equation}}
\newcommand{\ee}{\end{equation}}
\newcommand{\bea}{\begin{eqnarray}}
\newcommand{\eea}{\end{eqnarray}}
\newcommand{\ba}{\begin{array}}
\newcommand{\ea}{\end{array}}

\newcommand{\non}{\nonumber}

\newcommand{\sss}{\scriptscriptstyle}

\allowdisplaybreaks
\title{Joint Symbol-Level Precoding and Reflecting Designs for IRS-Enhanced MU-MISO Systems
\thanks{Part of this paper has been presented in the IEEE Wireless Communications and Networking Conference (WCNC), 2020 \cite{Liu WCNC 20}.}
\thanks{Manuscript received March 18, 2020; revised June 22, 2020; accepted September 28, 2020.
This work is supported by the National Natural Science Foundation of China (Grant No. 61671101, 61971088, and 61761136019), the Natural Science Foundation of Liaoning Province (Grant No. 20180520019), the Fundamental Research Funds for the Central Universities (Grant No. DUT20GJ214), and the National Science Foundation (Grants No. CCF-2008724 and ECCS-2030039).
The associate editor coordinating the review of this manuscript and approving it for publication was Dr. Wei Zhang. (\textit{Corresponding author: Ming Li}.)}
\thanks{R. Liu and M. Li are with the School of Information and Communication Engineering, Dalian University of Technology, Dalian 116024, China (e-mail: liurang@mail.dlut.edu.cn; mli@dlut.edu.cn).}
\thanks{Q. Liu is with the School of Computer Science and Technology, Dalian University of Technology, Dalian 116024, China (e-mail: qianliu@dlut.edu.cn).}
\thanks{A. L. Swindlehurst is with the center for Pervasive Communications and Computing, University of California, Irvine, CA 92697, USA (e-mail: swindle@uci.edu).}
}

\author{Rang Liu,~\IEEEmembership{Student Member,~IEEE,}
        Ming Li,~\IEEEmembership{Senior Member,~IEEE,}
        Qian Liu,~\IEEEmembership{Member,~IEEE,}
        and A. Lee Swindlehurst,~\IEEEmembership{Fellow,~IEEE}}

\begin{document}

\maketitle

\pagestyle{empty}
\thispagestyle{empty}

\begin{abstract}
Intelligent reflecting surfaces (IRSs) have emerged as a revolutionary solution to enhance wireless communications by changing propagation environment in a cost-effective and hardware-efficient fashion.
In addition, symbol-level precoding (SLP) has attracted considerable attention recently due to its advantages in converting multiuser interference (MUI) into useful signal energy.
Therefore, it is of interest to investigate the employment of IRS in symbol-level precoding systems to exploit MUI in a more effective way by manipulating the multiuser channels.
In this paper, we focus on joint symbol-level precoding and reflecting designs in IRS-enhanced multiuser multiple-input single-output (MU-MISO) systems.
Both power minimization and quality-of-service (QoS) balancing problems are considered.
In order to solve the joint optimization problems, we develop an efficient iterative algorithm to decompose them into separate symbol-level precoding and block-level reflecting design problems.
An efficient gradient-projection-based algorithm is utilized to design the symbol-level precoding and a Riemannian conjugate gradient (RCG)-based algorithm is employed to solve the reflecting design problem.
Simulation results demonstrate the significant performance improvement introduced by the IRS and illustrate the effectiveness of our proposed algorithms.

\end{abstract}


\maketitle

\section{Introduction}
\vspace{0.2 cm}

During the past decade, the applications of wireless communications have been growing rapidly and now affect nearly every aspect of our daily life.
Meanwhile, the demands of wireless communication networks for high data rate and low latency are also continuously increasing.
Various technical solutions have been proposed to meet the requirements of fifth-generation (5G) networks and beyond.
Among those technologies, massive multi-input multi-output (MIMO), millimeter wave (mmWave) communications, and ultra-dense networks are deemed as three fundamental approaches to enhance the performance along three basic dimensions: Improving spectral efficiency, utilizing more spectrum, and exploiting spatial reuse \cite{Zhang CST 2017}, \cite{Lee CM 2014}.
However, it seems that the performance improvements offered by these approaches are reaching their limits, and new technologies in different directions are needed to achieve further fundamental advances in wireless networks.
One such technology is the use of intelligent reflecting surfaces (IRS), which is a potentially revolutionary approach that provides additional degrees of freedom in system design by intelligently changing the propagation environment \cite{Liaskos CM 2018}-\cite{Wu CM 2020}.

An IRS is a planar array composed of a large number of reconfigurable passive reflecting elements, which are made up of some hardware-efficient devices, e.g., positive-intrinsic-negative (PIN) diodes and phase shifters.
Each reflecting element can independently and intelligently manipulate the amplitude and/or phase of incident electromagnetic (EM) waves in a programmable manner, which can produce a favorable propagation environment, especially when faced with blockages or severe fading.
New research on micro-electrical-mechanical systems (MEMS) and meta-materials has enabled the IRS to be configured in real-time, which is necessary for the rapidly changing wireless communication environment.
Thus, IRS have the potential for greatly expanding coverage, improving transmission quality, and enhancing security, etc., in a cost-effective and hardware-efficient fashion.
Moreover, these lightweight devices can be easily attached to the surfaces of buildings or mobile equipment, which provides mobility and portability for practical implementation \cite{Renzo 2019}, \cite{Basar Access 2019}.

Attracted by above advantages, researchers have devoted considerable attention to the development of IRS in the past year.
The applications of IRS to different wireless systems have been investigated to enhance their performance with different performance goals \cite{Zhao 2019}-\cite{Badiu WCL 2020}.
By properly designing the phase-shifting components of the IRS, the reflected signals can be coherently added to the received signals from other paths at intended receivers, which facilitates minimization of the transmit power \cite{Zhao 2019}, \cite{Wu TWC 2019}, or improving transmission performance in terms of spectral efficiency \cite{Yu ICCC 2019}, \cite{Zappone 2020}, energy reception \cite{Wu WCL 2019}, ergodic capacity \cite{Han TVT 2018}, symbol error rate \cite{Ye 2019}, channel capacity \cite{Perovic 2019}, sum-rate \cite{Huang 2018 ICASSP}-\cite{Li WCNC 2020}, and power efficiency \cite{Huang TWC 2019}, etc.
IRS-enhanced physical layer security has been investigated in \cite{Cui WCL 2019}-\cite{Yu GLOBECOM 2019}.
In order to provide preferable performance, the required number of reflecting elements has been studied in \cite{Bjornson WCL 2020}, \cite{Zhang TVT 2020}.
For practical implementation, researchers have investigated the applications of IRS with limited-resolution phase shifters \cite{Wu ITC 20}-\cite{Di TVT 2020} and phase errors \cite{Badiu WCL 2020}.
The IRS technique was also employed to realize index modulation \cite{Basar ITC 2020}, passive information transmission \cite{Yan JSAC 2020}, \cite{Liu WCSP 2019}, and holographic massive MIMO \cite{Huang 2019}.
Advanced machine learning-based algorithms have also been utilized to solve optimization problems related to IRS-assisted systems \cite{Yang 2020}, \cite{Huang 2020}.
In addition, practical channel estimation algorithms for IRS-assisted wireless communication systems have been proposed in \cite{Chen 2019}-\cite{Wei 2020}.

Precoding design is also of significant importance to facilitate information transmissions in IRS-enhanced multi-user systems.
In existing works, multi-user interference (MUI) is regarded as a harmful component and suppressed by the precoding and reflecting designs as much as possible.
However, recent research \cite{MA ITSP 2015}-\cite{Li ICST 2020} has found that MUI can often be treated as a source of useful signal energy to enhance information transmissions by means of symbol-level precoding techniques.
Specifically, symbol-level precoding utilizes transmitted symbol information and channel state information (CSI) to design the precoder, which converts harmful MUI into constructive interference to improve the symbol detection performance compared with linear  block-level precoding.

Motivated by these findings, we propose to combine symbol-level precoding and IRS in order to enjoy the advantages of both technologies.
The employment of IRS can facilitate the exploitation of MUI in symbol-level precoding by favorably manipulating the multi-user propagation channels.
However, there are some obstacles that must be tackled.
First, the symbol-level precoder changes with each transmitted symbol vector, while the IRS reflects all of them with the same phase-shift beamforming.
Thus, symbol-level constraints are difficult to implement in the reflecting design.
Since both the symbol-level precoding and reflecting need to consider all possible transmit vectors, the computational complexity will be tremendously high for large-scale systems and high-level modulation types.
To the best of our knowledge, this problem has not been studied yet, which motivates the work in this paper.

We consider the joint design of symbol-level precoding and IRS transmission design in multi-user multi-input single-output (MU-MISO) systems.
In particular, we consider a multi-antenna base station (BS) serving a number of single-antenna users with the aid of an IRS, which consists of many reflecting elements.
Our goal is to design the symbol-level precoder and IRS reflection to enhance the system performance by exploiting both MUI and modifications to the propagation environment.
The main contributions in this paper are summarized as follows:
\begin{itemize}
  \item We first focus on power minimization problem, which attempts to minimize the average transmit power as well as guarantee a certain quality-of-service (QoS) for the information transmissions. In order to solve this joint design problem, an efficient iterative algorithm is proposed to decompose the problem into separate symbol-level precoding and reflecting designs, where the gradient-projection-based and Riemannian conjugate gradient (RCG)-based algorithms are exploited.
  \item Then, we investigate the QoS balancing problem, which aims at maximizing the minimum QoS among users with a given average power constraint. The symbol-level precoding and reflecting are iteratively updated using similar gradient-projection-based and RCG-based algorithms after some transformations.
  \item Finally, we provide extensive simulation results to evaluate the performance of the proposed IRS-enhanced systems and the effectiveness of proposed algorithms. In particular, we show that applying IRS to symbol-level precoding MU-MISO systems brings remarkable performance improvements in terms of power-savings and symbol error rate (SER)-reduction, which illustrates how IRS have a positive effect on interference exploitation as well as the symbiotic benefits of using IRS and symbol-level precoding together.
\end{itemize}

The rest of this paper is organized as follows.
Section \ref{sec:system model} introduces the system model of our proposed IRS-enhanced MU-MISO system.
The considered power minimization and QoS balancing problems are investigated in Sections \ref{sec:PM} and \ref{sec:QoS}, respectively.
The algorithm initialization and complexity analysis are presented in Section \ref{sec:initialization}.
Simulation results are presented in Section \ref{sec:simulation}, and finally conclusions are provided in Section \ref{sec:conclusion}.

The following notation is used throughout this paper.
Boldface lower-case and upper-case letters indicate column vectors and matrices, respectively.
$(\cdot)^T$ and $(\cdot)^H$ denote the transpose and the transpose-conjugate operations, respectively.
$\mathbb{C}$ denotes the set of complex numbers.
$| a |$ and $\| \mathbf{a} \|$ are the magnitude of a scalar $a$ and the norm of a vector $\mathbf{a}$, respectively.
$\angle{a}$ is the angle of complex-valued $a$.
$\mathfrak{R}\{\cdot\}$ and $\mathfrak{I}\{\cdot\}$ denote the real and imaginary part of a complex number, respectively.
$\mathrm{diag}\{\mathbf{a}\}$ indicates the diagonal matrix whose diagonals are the elements of $\mathbf{a}$. $\mathbf{A} \succeq \mathbf{0}$ indicates that the matrix $\mathbf{A}$ is positive semi-definite.
Finally, we adopt the following indexing notation: $\mathbf{A}(i,j)$ denotes the element of the $i$-th row and the $j$-th column of matrix $\mathbf{A}$, and $\mathbf{a}(i)$ denotes the $i$-th element of vector $\mathbf{a}$.

\section{System Model}
\label{sec:system model}
\vspace{0.2 cm}

\begin{figure}[!t]
\center
\includegraphics[width = 0.5\textwidth]{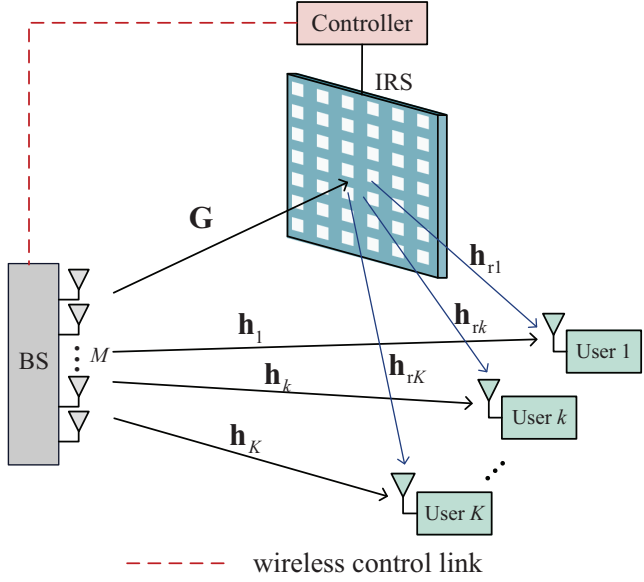}
\vspace{-0.4 cm}
\caption{The IRS-enhanced MU-MISO system.}
\label{fig:system model}
\vspace{-0.2 cm}
\end{figure}

We consider an IRS-enhanced MU-MISO system as shown in Fig. \ref{fig:system model}, where a BS equipped with $M$ antennas serves $K$ single-antenna users with the aid of an IRS.
The IRS consists of $N$ passive reflecting elements, which are implemented by phase shifters and denoted as $\bm{\theta} \triangleq \left[\theta_1,\ldots,\theta_N\right]$ that satisfy\footnote{In this preliminary work, we adopt the popular ideal phase-shift model \cite{Zhao 2019}-\cite{Zhang TVT 2020} which indicates a constant amplitude for all phase shifts.
In actual implementations, the reflection amplitudes will vary with the value of the phase shifts \cite{Abeywickrama 2020}, and these variations are associated with the difference between the frequency of the incident wave and the resonant frequency of the reflecting circuit \cite{Cai 2020}. Research assuming more realistic models will be conducted in future investigations.} $\left|\theta_n\right| = 1, \forall n$.
Each reflecting element is adaptively adjusted by the controller, which receives information about the optimized phase shifts from the BS through a dedicated control link.
We denote $\mathbf{G} \in \mathbb{C}^{N \times M}$, $\mathbf{h}_k \in \mathbb{C}^{M \times 1}$, and $\mathbf{h}_{\mathrm{r}k} \in \mathbb{C}^{N \times 1}$ as the channels from BS to IRS, from BS to the $k$-th user and from IRS to the $k$-th user, respectively.
In this paper, we assume that the CSI of all channels is known perfectly and instantaneously to the BS\footnote{In practice, channel estimation for IRS-enhanced systems is challenging since the passive reflecting elements in general do not have the ability to sense and process the received signals.
Although some initial channel estimation approaches \cite{Chen 2019}-\cite{Wei 2020} have been proposed, CSI is still difficult to obtain with limited overhead.
In order to focus on the impact of joint symbol-level precoding and reflecting problems, we assume perfect CSI in this paper.
Problems involving imperfect CSI or no CSI are left for future studies.}.

To facilitate the symbol-level precoding technique, we assume the transmitted symbols are independently selected from a $\Omega$-phase shift keying (PSK) constellation\footnote{We should emphasize that symbol-level precoding is related to the modulation type. Therefore, our designs in this paper are only capable of exploiting PSK modulation. The design for QAM modulation will be investigated in future work.} (i.e., $\Omega = 2, 4, \ldots$).
Therefore, the transmitted symbol vector $\mathbf{s}_m \triangleq \left[s_{m,1},\ldots,s_{m,K}\right]$ has $\Omega^K$ combinations, i.e., $m = 1,\ldots, \Omega^K$.
For different $\mathbf{s}_m$, the BS changes its transmitted precoder vector $\mathbf{x}_m \in \mathbb{C}^{M \times 1}$ in order to exploit the MUI.
Unlike conventional linear block-level precoding techniques, the mapping from $\mathbf{s}_m$ to $\mathbf{x}_m$ is usually nonlinear.

Through the direct and reflected paths, the compound received signal at the $k$-th user can be written as
\begin{equation}
r_{m,k} = \left(\mathbf{h}^H_k+\mathbf{h}^H_{\mathrm{r}k}\bm{\Theta}\mathbf{G}\right)\mathbf{x}_m+n_k, \forall m,
\end{equation}
where $\bm{\Theta} \triangleq \mathrm{diag}\left\{\bm{\theta}\right\}$ denotes the phase shifts of the IRS, and $n_k \sim \mathcal{CN}(0,\sigma_k^2)$ is additive white Gaussian noise (AWGN) for the $k$-th user.
Moreover, it is noteworthy that, during a coherent time slot with the same CSI, the BS changes $\mathbf{x}_m$ according to the transmitted symbols while the IRS phase shifts $\bm{\theta}$ remain unchanged.
Therefore, the reflecting design should consider all possible $\Omega^K$ precoding vectors.
Defining the precoding matrix $\mathbf{X} \triangleq \left[\mathbf{x}_1,\ldots,\mathbf{x}_{\Omega^K}\right]$, the average transmit power required to send a given symbol vector is
\begin{equation}
P_{\mathrm{ave}} = \frac{\left\|\mathbf{X}\right\|^2_{\sss{F}}}{\Omega^K}.
\end{equation}

\begin{figure}[!t]
\centering
\subfigure[An example of the constructive region.]{
\begin{minipage}{7 cm}
\centering
\includegraphics[width = 0.8\textwidth]{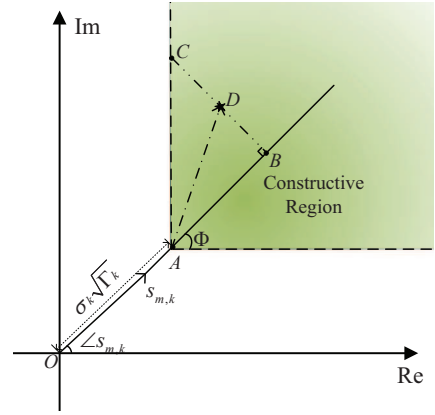}
\vspace{0.3 cm}
\label{fig:CR1}
\end{minipage}
}
\vspace{0.2 cm}
\subfigure[After rotating the diagram in Fig. \ref{fig:CR1} clockwise by $\angle s_{m,k}$ degrees.]{
\begin{minipage}{7 cm}
\centering
\includegraphics[width = 0.63\textwidth]{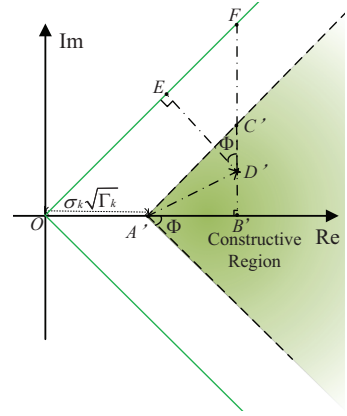}
\vspace{0.3 cm}
\label{fig:CR2}
\end{minipage}
}
\caption{Constructive region for a QPSK symbol.}
\label{fig:CR}
\vspace{-0.2 cm}
\end{figure}

With knowledge of the symbols to be transmitted, symbol-level precoding makes the MUI constructive to the information transmissions by elaborately designing $\mathbf{x}_m$.
In particular, the MUI is converted into constructive interference if it can push the received signals away from the PSK decision boundaries.
In order to illustrate the idea behind constructive interference, without loss of generality, we take quadrature phase shift keying (QPSK) as an example ($\Omega=4$).
For simplicity, we also consider $\left(\frac{1}{\sqrt{2}},\frac{1}{\sqrt{2}}j\right)$ as the symbol of interest for the $k$-th user and show the received signal in the complex plane as in Fig. \ref{fig:CR1}.
Since the decision boundaries for this symbol of interest are the positive halves of the $x$ and $y$ axes, as long as the noise-corrupted received signal expressed in (1) is in the first quadrant, the receiver can correctly detect the desired signal.
However, when we design the precoder $\mathbf{x}_m$, the noise is unknown and cannot be predicted beforehand.

Denote the received noise-free signal of the $k$-th user as \be \widetilde{r}_{m,k} = \left(\mathbf{h}^H_k+\mathbf{h}^H_{\mathrm{r}k}\bm{\Theta}\mathbf{G}\right)\mathbf{x}_m,\ee which is illustrated as point $D$ in Fig. \ref{fig:CR1}.
To reduce the impact of noise on the symbol detection, it is desirable to design $\mathbf{x}_m$ such that point $D$ is sufficiently far away from the corresponding decision boundaries to satisfy the QoS.
In order to quantify the QoS requirement, let $\Gamma_k$ be the signal-to-noise ratio (SNR) requirement for the $k$-th user.
If we ignore the impact of MUI and focus on the single-user case, the received noise-free signal should be at point $A$ to ensure that $\frac{|\widetilde{r}_{m,k}|^2}{\sigma_k^2}=\Gamma_k$, i.e., $\overrightarrow{OA}=\sigma_k\sqrt{\Gamma_k}\mathbf{s}_m(k)$.
Considering the multi-user signal as in (1), symbol-level precoding aims to design $\mathbf{x}_m$ such that point $D$ lies in the constructive (green) region, where the distance between received signal and its decision boundaries, which can be expressed as $\sigma_k\sqrt{\Gamma_k}\cos\Phi$, is guaranteed to satisfy the SNR requirement.
Therefore, symbol-level precoding can achieve better SER performance by converting the MUI into a constructive component.

In order to geometrically express this relationship, we project point $D$ on the direction of $\overrightarrow{OA}$ at point $B$, and define point $C$ to be the intersection of the extension of $\overrightarrow{BD}$ and the boundaries of the constructive region.
Then, the received noise-free signal in the green region should satisfy
\be
\left|\overrightarrow{BC}\right|-\left|\overrightarrow{BD}\right| \geq 0.
\ee
To make this expression clearer, we rotate the diagram in Fig. \ref{fig:CR1} clockwise by $\angle{s_{m,k}}$ degrees as shown in Fig. \ref{fig:CR2}.
Then, the QoS requirement is readily expressed as
\begin{equation}
\begin{aligned}
&\left[\mathfrak{R}\left\{\widetilde{r}_{m,k}e^{-j\angle{s_{m,k}}}\right\} - \sigma_k\sqrt{\Gamma_k}\right]\tan \Phi  \\
&~~~~~~~~~~~~~~~~~~~~-
\left|\mathfrak{I}\left\{\widetilde{r}_{m,k}e^{-j\angle{s_{m,k}}}\right\}\right| \geq 0,
\forall k, \forall m.
\end{aligned}
\end{equation}

In this paper, we consider two typical optimization problems for IRS-enhanced MU-MISO systems: \textit{i}) the power minimization problem, which minimizes the average transmit power while guaranteeing the QoS of the data received by the users; \textit{ii}) the QoS balancing problem, which aims to maximize the minimum QoS with a given average transmit power budget.
In the following sections, we will formulate and solve these two problems.

\section{Algorithm for Power Minimization Problem}
\label{sec:PM}
\vspace{0.2 cm}

With the previous analysis, the average power minimization problem can be written as
\begin{subequations}
\label{eq:orignal problem}
\begin{align}
  &\underset{\mathbf{X},\bm{\theta}}{\min}~~~\left\|\mathbf{X}\right\|_{\sss{F}}^2 \\
  &~\mathrm{s.t.}~~~\left[\mathfrak{R}\left\{\widetilde{r}_{m,k}e^{-j\angle{s_{m,k}}}\right\} - \sigma_k\sqrt{\Gamma_k}\right]\tan \Phi \\
  &~~~~~~~~~~~~~~~~- \left|\mathfrak{I}\left\{\widetilde{r}_{m,k}e^{-j\angle{s_{m,k}}}\right\}\right| \geq 0, \forall k, \forall m, \non\\
&~~~~~~~~~\widetilde{r}_{m,k} = \left(\mathbf{h}^H_k+\mathbf{h}^H_{\mathrm{r}k}\bm{\Theta}\mathbf{G}\right)\mathbf{x}_m, \forall k,\forall m, \\
\label{eq:problem d constraint}
&~~~~~~~~~\bm{\Theta}=\mathrm{diag}\{\bm{\theta}\}, \left|\theta_n\right|=1, \forall n,
\end{align}
\end{subequations}
which is non-convex due to the IRS constraints (\ref{eq:problem d constraint}).
Furthermore, the size of $\mathbf{X}\in \mathbb{C}^{M\times\Omega^K}$ is very large even for relatively small $K$ and $\Omega$.
Thus, it is difficult to directly solve this large-scale joint symbol-level precoding and reflecting problem.
In order to tackle this difficulty, we propose to decompose this bivariate problem into two sub-problems and implement the solutions iteratively.

\subsection{Symbol-Level Precoding Design for Power Minimization Problem}

When the IRS phase shifts $\bm{\theta}$ are fixed, the overall channel vector is determined.
We denote the combined channel vector from the BS to the $k$-th user as $\widetilde{\mathbf{h}}^H_k \triangleq \mathbf{h}^H_k + \mathbf{h}^H_{\mathrm{r}k}\bm{\Theta}\mathbf{G}, \forall k$.
Since the precoder vectors $\mathbf{x}_m, m = 1, \ldots, \Omega^K$, are independent of each other, the power minimization problem (\ref{eq:orignal problem}) can be divided into $\Omega^K$ sub-problems.
The $m$-th sub-problem for optimizing $\mathbf{x}_m$ is given by
\begin{subequations}
\label{eq:xm}
\begin{align}
  &\underset{\mathbf{x}_m}{\min}~~~\left\|\mathbf{x}_m\right\|^2 \\
  \label{eq: constraint 6b}
  &~\mathrm{s.t.}~~~\left[\mathfrak{R}\left\{\widetilde{\mathbf{h}}^H_k\mathbf{x}_me^{-j\angle{s_{m,k}}}\right\} - \sigma_k\sqrt{\Gamma_k}\right]\tan \Phi \\
  &~~~~~~~~~~~~~~~~-\left|\mathfrak{I}\left\{\widetilde{\mathbf{h}}^H_k\mathbf{x}_me^{-j\angle{s_{m,k}}}\right\}\right| \geq 0, \forall k. \non
\end{align}
\end{subequations}
This is a convex optimization problem and can be solved by standard convex tools, e.g.,  the CVX solver \cite{cvx}.
In addition, an efficient gradient projection algorithm with low complexity has also been studied in \cite{MA ITSP 2015}, which converts the problem to real-valued notation, and then derives its Lagrangian dual function to facilitate the gradient projection algorithm.
Due to space limitations, details of the algorithm to solve (\ref{eq:xm}) are omitted.

\subsection{Reflecting Design for Power Minimization Problem}

After obtaining the precoder vectors $\mathbf{x}_m, m = 1,\ldots,\Omega^K$, the objective of the original optimization problem (\ref{eq:orignal problem}) has been determined.
This means that, with given $\mathbf{x}_m$, the design of the reflection coefficients $\bm{\theta}$ becomes a feasibility-check problem whose outcome will not directly affect the power minimization objective of (\ref{eq:orignal problem}).
Therefore, for the reflecting design, we formulate another proper objective function that enhances the reduction of the transmit power for future iterations and guarantees its feasibility.

Since the power minimization problem (\ref{eq:xm}) is convex, the optimal $\mathbf{x}_m$ usually makes the left-hand side of constraint (\ref{eq: constraint 6b}) equal to a relatively small positive value, i.e., the QoS requirement is satisfied almost with equality.
In order to further reduce the transmit power in the next iteration, we propose to design the IRS phase shifts $\bm{\theta}$ using the stricter constraint (\ref{eq:stricter QoS}) in place of (\ref{eq: constraint 6b}), which can introduce an improved QoS that can provide more freedom for power minimization in the next iteration.
To this end, the IRS reflecting design problem is transformed to
\begin{subequations}
\label{eq:sum theta}
\begin{align}
  &\underset{\bm{\theta},\alpha_{m,k}}{\max}~~~
  \sum_{m=1}^{\Omega^K}\sum_{k=1}^{K}\alpha_{m,k} \\
  \label{eq:stricter QoS}
  &~\mathrm{s.t.}~~~~\left[\mathfrak{R}\left\{\widetilde{r}_{m,k}e^{-j\angle{s_{m,k}}}\right\} - \sigma_k\sqrt{\Gamma_k}\right]\tan \Phi \\
  &~~~~~~~~~~~~~~~~~~~~- \left|\mathfrak{I}\left\{\widetilde{r}_{m,k}e^{-j\angle{s_{m,k}}}\right\}\right| \geq \alpha_{m,k}, \forall k, \forall m, \non\\
&~~~~~~~~~\widetilde{r}_{m,k} = \left(\mathbf{h}^H_k+\mathbf{h}^H_{\mathrm{r}k}\bm{\Theta}\mathbf{G}\right)\mathbf{x}_m, \forall k,\forall m, \\
&~~~~~~~~~\bm{\Theta}=\mathrm{diag}\{\bm{\theta}\}, \left|\theta_n\right|=1, \forall n,
\end{align}
\end{subequations}
where the auxiliary variable $\alpha_{m,k}$ can be viewed as the residual QoS requirement.
Then, by defining
\begin{subequations}\label{eq:a b r}\begin{align}
a_{m,k} &\triangleq \mathbf{h}^H_k\mathbf{x}_me^{-j\angle s_{m,k}}, \forall k, \forall m, \\
\mathbf{b}_{m,k} &\triangleq \mathrm{diag}\left\{\mathbf{h}^H_{\mathrm{r}k}\right\}\mathbf{G}\mathbf{x}_me^{-j\angle s_{m,k}}, \forall k, \forall m, \\
\widehat{r}_{m,k} &\triangleq a_{m,k}+\bm{\theta}^H\mathbf{b}_{m,k},
\end{align}
\end{subequations}
problem (\ref{eq:sum theta}) can be reformulated concisely as
\begin{subequations}
\label{eq:min sum neat}
\begin{align}
\label{eq:objective min sum}
&\underset{\bm{\theta}}{\min}~~~\sum_{m=1}^{\Omega^K}\sum_{k=1}^{K}\left|\mathfrak{I}\left\{\widehat{r}_{m,k}\right\}\right|
-\left[\mathfrak{R}\left\{\widehat{r}_{m,k}\right\}-\sigma_k\sqrt{\Gamma_k}\right]\tan\Phi\\
\label{eq:unit modulus min sum}
&~\mathrm{s.t.}~~~\left|\theta_n\right| = 1, \forall n.
\end{align}
\end{subequations}
Unfortunately, the absolute values in the objective (\ref{eq:objective min sum}) are non-differentiable, which hinders the algorithm development.
In addition, the unit modulus constraints for the IRS phase shifts in (\ref{eq:unit modulus min sum}) introduce another difficulty due to their non-convexity.
Thus, we turn to solving above two problems by the log-sum-exp and manifold-based algorithms in the followings.

In order to handle the absolute value terms, we approximate the objective (\ref{eq:objective min sum}) by a differentiable function.
It can be observed that (\ref{eq:objective min sum}) can be concisely expressed as $|a|+b$, where $a$ and $b$ are scalars.
Then, the non-differentiable absolute value function can be replaced by
\begin{equation}
\label{eq:absolute}
\left|a\right|+b = \mathrm{max}\left\{a+b,-a+b\right\}.
\end{equation}
We then exploit the well-known log-sum-exp method and obtain
\begin{equation}
\begin{aligned}
&\mathrm{max}\left\{a+b,-a+b\right\} \\
&~~~~~~~~~~~~\approx \varepsilon \log \left[\exp\left(\frac{a+b}{\varepsilon}\right)+\exp\left(\frac{-a+b}{\varepsilon}\right)\right],
\end{aligned}
\end{equation}
where $\varepsilon$ is a relatively small positive number to maintain the approximation.
Thus, the optimization problem (\ref{eq:min sum neat}) can be reformulated as
\begin{subequations}
\label{eq:min sum smooth}
\begin{align}
\label{eq:obj min sum smooth}
&\underset{\bm{\theta}}{\min}~~~g \triangleq\sum_{i=1}^{K\Omega^K} \varepsilon\log\left[\exp\left(\frac{f_{2i-1}}{\varepsilon}\right)+\exp\left(\frac{f_{2i}}{\varepsilon}\right)\right] \\
\label{eq:unit min sum smooth}
&~\mathrm{s.t.}~~\left|\theta_n\right| = 1, \forall n.
\end{align}
\end{subequations}
For simplicity, in (\ref{eq:min sum smooth}) we define $f_i, i = 1, 2, \ldots, K\Omega^K$, as
\begin{subequations}
\label{eq: f2i}
\begin{align}
\centering
f_{2i-1} &\triangleq \mathfrak{I}\left\{\widehat{r}_{m,k}\right\} -  \left[\mathfrak{R}\left\{\widehat{r}_{m,k}\right\}-\sigma_k\sqrt{\Gamma_k}\right]\tan\Phi \non \\
&= \mathfrak{R}\{\bm{\theta}^H\}\mathbf{a}_{2i-1} + \mathfrak{I}\{\bm{\theta}^H\}\mathbf{b}_{2i-1}+c_{2i-1}, \\
f_{2i}~~~&\triangleq -\mathfrak{I}\left\{\widehat{r}_{m,k}\right\}-\left[\mathfrak{R}\left\{\widehat{r}_{m,k}\right\}-\sigma_k\sqrt{\Gamma_k}\right]\tan\Phi \non \\
&= \mathfrak{R}\{\bm{\theta}^H\}\mathbf{a}_{2i} + \mathfrak{I}\{\bm{\theta}^H\}\mathbf{b}_{2i}+c_{2i},
\end{align}
\end{subequations}
where $i = K(m-1)+k$ and
\begin{subequations}
\label{eq:abc}
\begin{align}
\label{eq:a 2i-1}
\mathbf{a}_{2i-1} &\triangleq \mathfrak{I}\{\mathbf{b}_{m,k}\}-\mathfrak{R}\{\mathbf{b}_{m,k}\}\tan\Phi,\\
\label{eq:b 2i-1}
\mathbf{b}_{2i-1} &\triangleq \mathfrak{R}\{\mathbf{b}_{m,k}\}+\mathfrak{I}\{\mathbf{b}_{m,k}\}\tan\Phi,\\
c_{2i-1} &\triangleq \mathfrak{I}\{a_{m,k}\}-\mathfrak{R}\{a_{m,k}\}\tan\Phi+\sigma_k\sqrt{\Gamma_k}\tan\Phi,\\
\label{eq:a 2i}
\mathbf{a}_{2i}~~~ &\triangleq -\mathfrak{I}\{\mathbf{b}_{m,k}\}-\mathfrak{R}\{\mathbf{b}_{m,k}\}\tan\Phi, \\
\label{eq:b 2i}
\mathbf{b}_{2i}~~~ &\triangleq -\mathfrak{R}\{\mathbf{b}_{m,k}\}+\mathfrak{I}\{\mathbf{b}_{m,k}\}\tan\Phi, \\
c_{2i}~~~ &\triangleq -\mathfrak{I}\{a_{m,k}\}-\mathfrak{R}\{a_{m,k}\}\tan\Phi+\sigma_k\sqrt{\Gamma_k}\tan\Phi.
\end{align}
\end{subequations}

While the objective of (\ref{eq:min sum smooth}) is smooth and differentiable, the unit modulus constraints (\ref{eq:unit min sum smooth}) are non-convex, which still makes the problem difficult to solve.
Two popular methods for handling this type of constraint include non-convex relaxation and alternating minimization.
However, the non-convex relaxation method always suffers a performance loss since the solution is based on a relaxation of the original problem.
On the other hand, the alternating minimization method may have slow convergence due to the large number of variables involved.
In order to deal with these difficulties, we adopt the Riemannian-manifold-based algorithm, which can achieve a locally optimal solution of the original optimization problem with very fast convergence \cite{Boumal manopt 14}.

Before developing the algorithm, we need to introduce some related concepts.
On a manifold, each point has a neighborhood homeomorphic to Euclidean space, and the directions in which the point can move are its tangent vectors, which compose the tangent space.
Similar to the Euclidean space, the tangent space has one tangent vector in the direction where the objective function decreases fastest, which is referred to as the Riemannian gradient.
Furthermore, the Riemannian gradient is the orthogonal projection of the Euclidean gradient onto its corresponding tangent space.
Therefore, efficient algorithms used in Euclidean space, e.g., the conjugate gradient (CG) and the trust-region methods, are suitable on the Riemannian manifold after several operations.
In the following, we apply the conjugate gradient algorithm on the Riemannian manifold to solve our problem.

Denoting $\widetilde{\bm{\Theta}} \triangleq \left[\mathfrak{R}\{\bm{\theta}\},\mathfrak{I}\{\bm{\theta}\}\right]^T$, the unit modulus constraints (\ref{eq:unit min sum smooth}) form a $2N$-dimensional smooth Riemannian manifold
\begin{equation}
\label{eq:search space}
\mathcal{M} = \left\{\widetilde{\bm{\Theta}} \in \mathbb{R}^{2 \times N}:[\widetilde{\bm{\Theta}}(:,n)]^T\widetilde{\bm{\Theta}}(:,n) = 1, \forall n\right\},
\end{equation}
whose tangent space is
\begin{equation}
T_{\widetilde{\bm{\Theta}}}\mathcal{M} = \left\{\mathbf{P} \in \mathbb{C}^{2\times N}:[\widetilde{\bm{\Theta}}(:,n)]^T\mathbf{P}(:,n) = 0, \forall n\right\}.
\end{equation}
In order to facilitate the conjugate gradient algorithm, the Euclidean gradient is required to determine the corresponding Riemannian gradient.
Let $\widetilde{\bm{\theta}}_n$ be the $n$-th column of $\widetilde{\bm{\Theta}}$, so that the Euclidean gradient of $g(\widetilde{\bm{\Theta}})$ can be expressed as
\begin{equation}
\begin{aligned}
\nabla_{\widetilde{\bm{\Theta}}} g = \left[\frac{\partial g}{\partial\widetilde{\bm{\theta}}_1},\ldots,\frac{\partial g}{\partial\widetilde{\bm{\theta}}_N}\right].
\end{aligned}
\end{equation}
Following the chain rule, the $n$-th column of the Euclidean gradient is calculated as
\begin{equation}
\label{eq:chain rule}
\begin{aligned}
\frac{\partial g}{\partial\widetilde{\bm{\theta}}_n} =\frac{\partial \mathfrak{R}\{\bm{\theta}^H\}}{\partial\widetilde{\bm{\theta}}_n}\left(\frac{\partial g}{\partial\mathfrak{R}\{\bm{\theta}^H\}}\right)^T +\frac{\partial \mathfrak{I}\{\bm{\theta}^H\}}{\partial\widetilde{\bm{\theta}}_n}\left(\frac{\partial g}{\partial\mathfrak{I}\{\bm{\theta}^H\}}\right)^T.
\end{aligned}
\end{equation}

According to the previous definition, it is obvious that
\begin{equation}
\label{eq: gradient e}
\begin{aligned}
\frac{\partial \mathfrak{R}\{\bm{\theta}^H\}}{\partial\widetilde{\bm{\theta}}_n} = \left[\mathbf{e}_n,\mathbf{0}\right]^T, \frac{\partial \mathfrak{I}\{\bm{\theta}^H\}}{\partial\widetilde{\bm{\theta}}_n} = \left[\mathbf{0},\mathbf{e}_n\right]^T,
\end{aligned}
\end{equation}
where $\mathbf{e}_n \in \mathbb{R}^{N \times 1}$ is defined by $\mathbf{e}_n(n) = 1, \mathbf{e}_n(i) = 0, \forall i \neq n$.
Based on (\ref{eq:obj min sum smooth}) and (\ref{eq: f2i}), we have
\begin{subequations}\label{eq:drivate ri}\begin{align}
&\frac{\partial g}{\partial\mathfrak{R}\{\bm{\theta}^H\}} = \sum_{i = 1}^{K\Omega^K}\frac{\exp(f_{2i-1}/\varepsilon)\mathbf{a}_{2i-1}^T+\exp(f_{2i}/\varepsilon)\mathbf{a}_{2i}^T}
{\exp(f_{2i-1}/\varepsilon)+\exp(f_{2i}/\varepsilon)}, \\
&\frac{\partial g}{\partial\mathfrak{I}\{\bm{\theta}^H\}} = \sum_{i = 1}^{K\Omega^K}\frac{\exp(f_{2i-1}/\varepsilon)\mathbf{b}_{2i-1}^T+\exp(f_{2i}/\varepsilon)\mathbf{b}_{2i}^T}
{\exp(f_{2i-1}/\varepsilon)+\exp(f_{2i}/\varepsilon)}.
\end{align}\end{subequations}
Then, the Euclidean gradient can be readily calculated by substituting (\ref{eq: gradient e}) and (\ref{eq:drivate ri}) into (\ref{eq:chain rule}).
The Riemannian gradient is thus given by
\begin{equation}
\label{eq:grad}
\mathrm{grad}_{\widetilde{\bm{\Theta}}}g = \mathcal{P}_{\widetilde{\bm{\Theta}}}\left(\nabla_{\widetilde{\bm{\Theta}}} g\right) = \nabla_{\widetilde{\bm{\Theta}}} g - \widetilde{\bm{\Theta}} \mathrm{diag}\left\{\widetilde{\bm{\Theta}}^T\nabla_{\widetilde{\bm{\Theta}}} g \right\},
\end{equation}
where $\mathcal{P}_{\widetilde{\bm{\Theta}}}\left(\cdot\right)$ denotes the projection onto the tangent space.

With the Riemannian gradient, the conjugate gradient algorithm can be employed on the Rienmannian space, and is referred to as the Riemannian conjugate gradient (RCG) algorithm.
Considering the characteristics of the Riemannian space, this line search method works in a different way than the standard CG algorithm.
In the $p$-th iteration of RCG, the search direction $\mathbf{d}_p$ is determined by the Riemannian gradient $\mathrm{grad}_{\widetilde{\bm{\Theta}}}g(\widetilde{\bm{\Theta}}_p)$ and the $(p-1)$-th search direction $\mathbf{d}_{p-1}$.
Since these two vectors lie in different tangent spaces, they cannot be directly combined.
Thus, the Riemannian transport operation is needed to map $\mathbf{d}_{p-1}$ into the tangent space of $\mathrm{grad}_{\widetilde{\bm{\Theta}}}g(\widetilde{\bm{\Theta}}_p)$.
Then, the search direction $\mathbf{d}_p$ is given by
\begin{equation}
\label{eq:dp}
\mathbf{d}_p = -\mathrm{grad}_{\widetilde{\bm{\Theta}}}g(\widetilde{\bm{\Theta}}_p) + \beta_p \mathbf{d}_{p-1}^\mathrm{t},
\end{equation}
where $\beta_p$ is the Polak-Ribiere parameter \cite{CG} and the superscript ``t'' indicates the Riemannian transport operation.
The step size $\alpha_p$ is chosen by the Armijo backtracking line search method \cite{CG} and the $p$-th update is thus expressed as
\begin{equation}
\label{eq:retr}
\widetilde{\bm{\Theta}}_p = \mathrm{Retr}_{\widetilde{\bm{\Theta}}}\left(\widetilde{\bm{\Theta}}_{p-1}+\alpha_p\mathbf{d}_p\right),
\end{equation}
where $\mathrm{Retr}_{\widetilde{\bm{\Theta}}}\left(\cdot\right)$ indicates the retraction operation, which maps the point on the tangent space to the manifold.

After obtaining $\widetilde{\bm{\Theta}}^*$, the optimal $\bm{\theta}^*$ can be constructed as
\begin{equation}
\label{eq:theta o}
\bm{\theta}^* = [\widetilde{\bm{\Theta}}^*(1,:)]^T+j[\widetilde{\bm{\Theta}}^*(2,:)]^T.
\end{equation}
The RCG algorithm to obtain $\bm{\theta}^*$ is summarized in Algorithm \ref{alg:RCG}, where $N_\mathrm{max}$ is the maximum number of iterations and $\delta_\mathrm{th}$ is the threshold to judge convergence.

\begin{algorithm}[!t]
\caption{RCG-based IRS Reflecting Design}
\label{alg:RCG}
    \begin{algorithmic}[1]
    \REQUIRE $g(\widetilde{\bm{\Theta}})$, $\widetilde{\bm{\Theta}}_0 \in \mathcal{M}$, $N_\mathrm{max}$, $\delta_{\mathrm{th}}$.
    \ENSURE $\bm{\theta}^*$.
        \STATE {Initialize $p = 0$, $\delta=\infty$, $\mathbf{d}_0 = -\mathrm{grad}_{\widetilde{\bm{\Theta}}}g(\widetilde{\bm{\Theta}}_0)$.}
        \WHILE {$p \leq N_{\mathrm{max}}$ and $\delta \geq \delta_{\mathrm{th}}$ }
            \STATE{Calculate Riemannian gradient $\mathrm{grad}_{\widetilde{\bm{\Theta}}}g(\widetilde{\bm{\Theta}}_{p})$ by (\ref{eq:grad}).}
            \STATE{Choose Polak-Ribiere parameter $\beta_{p}$ \cite{CG}.}
            \STATE{Calculate search direction $\mathbf{d}_{p}$ by (\ref{eq:dp}).}
            \STATE{Calculate Armijo backtracking line search step size $\alpha_p$ \cite{CG}.}
            \STATE{Obtain $\widetilde{\bm{\Theta}}_{p}$ by (\ref{eq:retr}).}
            \STATE{$\delta = \left\|\mathrm{grad}_{\widetilde{\bm{\Theta}}}g(\widetilde{\bm{\Theta}}_{p})\right\|^2$}
            \STATE {$p = p + 1$.}
        \ENDWHILE
        \STATE{$\widetilde{\bm{\Theta}}^* = \widetilde{\bm{\Theta}}_{p}$.}
        \STATE{Construct $\bm{\theta}^*$ by (\ref{eq:theta o})}
    \end{algorithmic}
\end{algorithm}
In realistic IRS implementations, low-resolution digital phase shifters are more hardware-efficient and practical.
A discrete phase-shift $\bm{\theta}_\mathrm{d}$ for the IRS using $B$ bits of resolution can thus be calculated by direct quantization, i.e.,
\begin{equation}
\label{eq:theta B}
\bm{\theta}_\mathrm{d}(n) = \mathrm{round}\left\{\frac{\bm{\theta}^*(n)}{2 \pi /2^B}\right\}\times \frac{2\pi}{2^B}, \forall n,
\end{equation}
where $\mathrm{round}\{\cdot\}$ indicates rounding to the nearest integer.

Now, with the previous developments, the joint symbol-level precoding and reflecting design for the power minimization problem is straightforward.
Given an initial value $\bm{\theta}_0$, the symbol-level precoders $\mathbf{x}_m, m = 1,\ldots, \Omega^K$, and the IRS phase shifts $\bm{\theta}$ are iteratively updated by solving (\ref{eq:xm}) and (\ref{eq:sum theta}) until convergence is found.
This joint symbol-level precoding and reflecting algorithm for the power minimization problem is summarized in Algorithm \ref{alg:2}.
Selection of an initial $\bm{\theta}_0$ will be addressed in Section \ref{sec:initialization}.
\begin{algorithm}[!t]
\caption{Joint Symbol-Level Precoding and Reflecting Design for the Power Minimization Problem}
\label{alg:2}
    \begin{algorithmic}[1]
    \REQUIRE $\mathbf{h}_k$, $\mathbf{h}_{\mathrm{r}k}$, $\mathbf{G}$, $\Omega$, $\sigma_k$, $\Gamma_k$, $B$, $N_{\mathrm{max}}$, $\delta_{\mathrm{th}}$.
    \ENSURE $\bm{\theta}^*$, $\mathbf{X}^*$.
        \STATE {Initialize $\bm{\theta}_0$ by solving (\ref{eq:initial theta 2}), $iter = 0$, $\delta=\infty$, $p_\mathrm{t} = 0$.}
        \WHILE {$iter \leq N_{\mathrm{max}}$ and $\delta \geq \delta_{\mathrm{th}}$ }
            \STATE{$p_\mathrm{pre} = p_\mathrm{t}$.}
            \STATE{Calculate precoder $\mathbf{x}_m, m = 1,\ldots, \Omega^K$, by (\ref{eq:xm}).}
            \STATE{Obtain infinite resolution IRS phase shifts $\bm{\theta}$ using Algorithm \ref{alg:RCG}.}
            \STATE{Calculate $\bm{\theta}_\mathrm{d}$ by (\ref{eq:theta B}) for the low-resolution IRS phase shifts cases.}
            \STATE{$p_\mathrm{t} = \left\|\mathbf{X}\right\|_{\sss{F}}^2$.}
            \STATE{$\delta = \left|\frac{p_\mathrm{t}-p_\mathrm{re}}{p_\mathrm{re}}\right|$.}
            \STATE {$iter = iter + 1$.}
        \ENDWHILE
        \STATE{$\bm{\theta}^* = \bm{\theta}$ or $\bm{\theta}_\mathrm{d}$, $\mathbf{X}^* = \mathbf{X}$.}
    \end{algorithmic}
\end{algorithm}

\section{Algorithm for QoS Balancing Problem}
\label{sec:QoS}
\vspace{0.2 cm}

In this section, we first formulate the QoS balancing problem for the considered IRS-enhanced MU-MISO system.
Then, a similar algorithm is proposed to iteratively solve the symbol-level precoding and reflecting design problems.

As discussed in Section \ref{sec:system model}, the distance between the received noise-free signal and its decision boundaries essentially determines symbol detection performance;
larger distances provide stronger robustness against noise, and thus a lower SER.
Therefore, we use this distance as the QoS metric and aim at maximizing the minimum QoS among the users with a given power budget.
From Fig. \ref{fig:CR}, we observe that the distance between point $D$ and its decision boundaries (i.e., the positive halves of the $x$ and $y$ axes in this case) can be expressed as
\be\begin{aligned}
\label{eq:distance}
&\left|\overrightarrow{D'E}\right|= \left|\overrightarrow{D'F}\right|\cos\Phi
= \left(\left|\overrightarrow{B'F}\right|-\left|\overrightarrow{B'D'}\right|\right)\cos\Phi \\
&= \left[\mathfrak{R}\{\widetilde{r}_{m,k}e^{-j\angle{s_{m,k}}}\}\tan \Phi - \left|\mathfrak{I}\{\widetilde{r}_{m,k}e^{-j\angle{s_{m,k}}}\}\right|\right]\cos \Phi.
\end{aligned}\ee
Thus, after ignoring the constant term $\cos \Phi$, the QoS balancing problem can be formulated as
\begin{subequations}
\label{eq:SINR balancing}
\begin{align}
\underset{\mathbf{X},\bm{\theta}}{\max}~~&\underset{m,k}{\min}~~~\mathfrak{R}\left\{\widetilde{r}_{m,k}
  e^{-j\angle{s_{m,k}}}\right\}
   \tan \Phi- \left|\mathfrak{I}\left\{\widetilde{r}_{m,k}e^{-j\angle{s_{m,k}}}\right\}\right|\\
  &~\mathrm{s.t.}~~~\widetilde{r}_{m,k} = \left(\mathbf{h}^H_k+\mathbf{h}^H_{\mathrm{r}k}\bm{\Theta}\mathbf{G}\right)\mathbf{x}_m, \forall k,\forall m, \\
\label{eq:SINR theta}
&~~~~~~~~\bm{\Theta}=\mathrm{diag}\{\bm{\theta}\}, \left|\theta_n\right|=1, \forall n,\\
&~~~~~~~~\left\|\mathbf{X}\right\|^2 \leq P\Omega^K,
\end{align}
\end{subequations}
where $P$ is the preset average transmit power budget.
As before, we propose to decompose this bivariate problem into separate symbol-level precoding design and the reflecting design problems, and solve them iteratively.

\subsection{Symbol-Level Precoding Design for QoS Balancing Problem}

With given IRS phase shifts $\bm{\theta}$, the combined channel vector from the BS to the $k$-th user is $\widetilde{\mathbf{h}}^H_k \triangleq \mathbf{h}^H_k + \mathbf{h}^H_{\mathrm{r}k}\bm{\Theta}\mathbf{G}$.
Then, the QoS balancing problem for designing the precoder $\mathbf{X}$ can be rewritten as
\begin{subequations}
\label{eq:SINR balancing precoding}
\begin{align}
  &\underset{\mathbf{X},t}{\max}~~~t \\
  &~\mathrm{s.t.}~~~\mathfrak{R}\{\widetilde{\mathbf{h}}^H_k
  \mathbf{x}_me^{-j\angle{s_{m,k}}}\}\tan \Phi\\
  &~~~~~~~~~~~~~~~~~~~~~-\left|\mathfrak{I}\{\widetilde{\mathbf{h}}^H_k\mathbf{x}_me^{-j\angle{s_{m,k}}}\}\right| \geq t, \forall k, \forall m,\non \\
  &~~~~~~~\left\|\mathbf{X}\right\|^2 \leq P\Omega^K,
\end{align}
\end{subequations}
which is a convex problem and can be solved by standard convex optimization tools, e.g.,  CVX.
However, since the variable to be optimized $\mathbf{X}$ has a large dimension of $M\Omega^K$, the complexity is unaffordable.
In order to deal with this difficulty, we decompose this problem into $\Omega^K$ sub-problems, where $\mathbf{x}_m, \forall m$, is individually designed.
To facilitate the algorithm development, we propose the following proposition.
\begin{prop}
Let $\mathbf{x}_1^*, \ldots, \mathbf{x}_{\Omega^K}^*$ be the optimal solution of the QoS balancing problem (\ref{eq:SINR balancing precoding}).
Let $\mathbf{x}_1^\star,\ldots,\mathbf{x}_{\Omega^K}^\star$ be the optimal solution of the power minimization problem (\ref{eq:xm}), where the QoS requirement for all users equals $t_0=\sigma_k\sqrt{\Gamma_k}\tan\Phi, \forall k$.
Then, $\mathbf{x}_m^*$ is a scaled version of $\mathbf{x}_m^\star$, i.e., $\mathbf{x}_m^* = \frac{\sqrt{P_m}\mathbf{x}_m^\star}{\left\|\mathbf{x}_m^\star\right\|}$, where $P_m \geq 0$ is the transmit power allocated to the $m$-th precoder, $\sum_{m=1}^{\Omega^K}P_m = P\Omega^K$.
Furthermore, the minimum QoS that $\mathbf{x}_m^*$ can achieve is $\frac{\sqrt{P_m}t_0}{\left\|\mathbf{x}_m^\star\right\|}$.
\end{prop}
\begin{proof}
See Appendix A.
\end{proof}
Proposition 1 indicates that we can first find the precoder $\mathbf{x}_m^\star$ by individually solving the power minimization problem (\ref{eq:xm}) with a given QoS requirement $t_0$, and then scaling $\mathbf{x}_m^\star$ appropriately to obtain the optimal $\mathbf{x}_m^*$ by finding the optimal power allocation $P_m$.
With given $\mathbf{x}_m^\star$, the power allocation problem to optimize QoS balancing can be formulated as
\begin{subequations}
\label{eq:power allocation}
\begin{align}
&\underset{P_m,\forall m}{\max}~~~t\\
&~\mathrm{s.t.}~~~~~t \leq \frac{\sqrt{P_m}t_0}{\left\|\mathbf{x}_m^\star\right\|}, \forall m, \\
&~~~~~~~~~\sum_{m=1}^{\Omega^K}P_m \leq P\Omega^K.
\end{align}
\end{subequations}
While (\ref{eq:power allocation}) is convex and can be solved by CVX, we attempt to find a more efficient solution to reduce the complexity.
Denoting $\mathbf{p} \triangleq \left[\sqrt{P_1},\ldots,\sqrt{P_{\Omega^K}}\right]^T$, the power allocation problem  (\ref{eq:power allocation}) can be rewritten as
\begin{subequations}
\label{eq:power allocate}
\begin{align}
&\underset{\mathbf{p},t}{\max}~~~t\\
&~\mathrm{s.t.}~~~t \leq \mathbf{e}_m^T\mathbf{p}, \forall m, \\
&~~~~~~~\left\|\mathbf{p}\right\|^2 \leq P\Omega^K,
\end{align}
\end{subequations}
where $\mathbf{e}_m$ is a vector of all zeros except the $m$-th element which is $ \frac{t_0}{\left\|\mathbf{x}_m^\star\right\|}$.
Motivated by Proposition~1, we first find $\mathbf{p}^\star$, which is the optimal solution for the following power minimization problem with an arbitrary given $t_0' \geq 0$:
\begin{subequations}
\label{eq:min p}
\begin{align}
&\underset{\mathbf{p}}{\min}~~~\left\|\mathbf{p}\right\|^2\\
&~\mathrm{s.t.}~~~~\mathbf{e}_m^T\mathbf{p} \geq t_0', \forall m.
\end{align}
\end{subequations}
This problem can be efficiently solved using the same method as in problem (\ref{eq:xm}), based on the Lagrangian dual problem and exploiting the gradient projection algorithm.
Then, the optimal $\mathbf{p}^*$ for (\ref{eq:power allocate}) can be obtained by $\mathbf{p}^* = \frac{\sqrt{P\Omega^K}\mathbf{p}^\star}{\left\|\mathbf{p}^\star\right\|}$.

With the above analysis, the symbol-level precoding algorithm for the QoS balancing problem can be summarized as:
\textit{i}) obtain the precoder $\mathbf{x}_m^\star$ by solving the power minimization problem (\ref{eq:xm}) with a certain QoS requirement $t_0 \geq 0$; \textit{ii}) solving the power allocation problem (\ref{eq:power allocation}) to obtain $P_m, \forall m$; \textit{iii}) scaling $\mathbf{x}_m^\star$ to obtain the optimal solution of (\ref{eq:SINR balancing precoding}) as $\mathbf{x}_m^* = \frac{\sqrt{P_m}\mathbf{x}_m^\star}{\left\|\mathbf{x}_m^\star\right\|}$.

\subsection{Reflecting Design for QoS Balancing Problem}

With fixed precoders $\mathbf{x}_1,\ldots,\mathbf{x}_{\Omega^K}$, the reflecting design problem is given by
\begin{subequations}
\label{eq:SINR balancing theta}
\begin{align}
\underset{\bm{\theta}}{\max}~~ &\underset{m,k}{\min}~~~\mathfrak{R}\{\widetilde{r}_{m,k}
  e^{-j\angle{s_{m,k}}}\}\tan \Phi- \left|\mathfrak{I}\{\widetilde{r}_{m,k}e^{-j\angle{s_{m,k}}}\}\right| \\
  &~\mathrm{s.t.}~~~\widetilde{r}_{m,k} = \left(\mathbf{h}^H_k+\mathbf{h}^H_{\mathrm{r}k}\bm{\Theta}\mathbf{G}\right)\mathbf{x}_m, \forall k,\forall m, \\
&~~~~~~~~\bm{\Theta}=\mathrm{diag}\{\bm{\theta}\}, \left|\theta_n\right|=1, \forall n.
\end{align}
\end{subequations}
Using the definitions in (\ref{eq:a b r}), the reflecting design problem is more compactly formulated as
\begin{subequations}
\begin{align}
\label{eq:trans SINR obj}
\underset{\bm{\theta}}{\min}~~ &\underset{m,k}{\max}~~~\left|\mathfrak{I}\left\{\widehat{r}_{m,k}\right\}\right|
-\mathfrak{R}\left\{\widehat{r}_{m,k}\right\}\tan\Phi \\
\label{eq:trans SINR theta}
  &~\mathrm{s.t.}~~~~\left|\theta_n\right|=1, \forall n.
\end{align}
\end{subequations}
As before, (\ref{eq:trans SINR obj}) is non-differentiable due to the max and absolute value functions and (\ref{eq:trans SINR theta}) is non-convex, which leads us to exploit the RCG algorithm.
To facilitate the RCG algorithm, the same idea used to solve (\ref{eq:objective min sum}) is employed here in three steps: \textit{i}) replacing the absolute value function, \textit{ii}) smoothing the max function, \textit{iii}) calculating its Euclidean gradient.
We briefly describe these three steps below.

The absolute value function is replaced using (\ref{eq:absolute}), and the problem is further rearranged as
\begin{subequations}
\begin{align}
\underset{\widetilde{\bm{\Theta}}}{\min}~~&\underset{i}{\max}~~~f_i \\
&~\mathrm{s.t.}~~~[\widetilde{\bm{\Theta}}(:,n)]^T\widetilde{\bm{\Theta}}(:,n) = 1, \forall n,
\end{align}
\end{subequations}
where $f_i$ is defined in (\ref{eq: f2i}) with the auxiliary vectors $\mathbf{a}_{2i-1}$,  $\mathbf{b}_{2i-1}$, $\mathbf{a}_{2i}$ $\mathbf{b}_{2i}$ in (\ref{eq:a 2i-1}), (\ref{eq:b 2i-1}), (\ref{eq:a 2i}), (\ref{eq:b 2i}), and
\begin{subequations}
\begin{align}
c_{2i-1} &\triangleq \mathfrak{I}\{a_{m,k}\}-\mathfrak{R}\{a_{m,k}\}\tan\Phi,\\
c_{2i} &\triangleq -\mathfrak{I}\{a_{m,k}\}-\mathfrak{R}\{a_{m,k}\}\tan\Phi.
\end{align}
\end{subequations}

We smooth the max function by exploiting the log-sum-exp algorithm, which introduces the approximation
\begin{equation}
\label{eq:SINR g}
\begin{aligned}
&\max\{f_1,f_2, \ldots, f_{2\Omega^K}\} \approx  g(\widetilde{\bm{\Theta}})\\
&~~~~\triangleq \varepsilon \log \left\{\sum_{i=1}^{K\Omega^K}\left[\exp\left(\frac{f_{2i-1}}{\varepsilon}\right)+\exp\left(\frac{f_{2i}}{\varepsilon}\right)\right]\right\},
\end{aligned}
\end{equation}
where $\varepsilon$ is a small positive number.

After obtaining the smooth and differentiable $g(\widetilde{\bm{\Theta}})$, its Euclidean gradient can be derived by substituting (\ref{eq: gradient e}) into (\ref{eq:chain rule}), where we need to calculate
\begin{subequations}
\begin{align}
\centering
&\frac{\partial g}{\partial\mathfrak{R}\{\bm{\theta}^H\}} = \frac{\sum_{i=1}^{K \Omega^K}\left[\exp(f_{2i-1}/\varepsilon) \mathbf{a}^T_{2i-1}+\exp(f_{2i}/\varepsilon) \mathbf{a}^T_{2i}\right]}{\sum_{i=1}^{K\Omega^K}\left[\exp(f_{2i-1}/\varepsilon)+\exp(f_{2i+1}/\varepsilon)\right]}, \\
&\frac{\partial g}{\partial\mathfrak{I}\{\bm{\theta}^H\}} = \frac{\sum_{i=1}^{K \Omega^K}\left[\exp(f_{2i-1}/\varepsilon) \mathbf{b}^T_{2i-1}+\exp(f_{2i}/\varepsilon) \mathbf{b}^T_{2i}\right]}{\sum_{i=1}^{K\Omega^K}\left[\exp(f_{2i-1}/\varepsilon)+\exp(f_{2i+1}/\varepsilon)\right]}.
\end{align}
\end{subequations}
Then, the RCG-based reflecting design in Algorithm \ref{alg:RCG} can be applied to solve the QoS balancing problem.
The optimal IRS phase shifts $\bm{\theta}^*$ and low-resolution phase shifts $\bm{\theta}_\mathrm{d}$ have the same format as in (\ref{eq:theta o}) and (\ref{eq:theta B}).

Finally, the joint symbol-level precoding and reflecting design for the QoS balancing problem is straightforward.
With an initial reflecting value $\bm{\theta}_0$, the symbol-level precoding matrix $\mathbf{X}$ and the IRS phase shifts $\bm{\theta}$ are iteratively updated by solving (\ref{eq:SINR balancing precoding}) and (\ref{eq:SINR balancing theta}) until convergence is found.

\section{Initialization and Complexity Analysis}
\label{sec:initialization}
\vspace{0.2 cm}
\subsection{Initialization}

Since the RCG algorithm in general will find a locally optimal solution, an initial value that is close to the optimal solution can provide better performance and accelerate convergence.
In this subsection, we propose a heuristic method to obtain the initial $\bm{\theta}_0$.

Both the power minimization and QoS balancing problems depend on the quality of the users' channels, which can be manipulated by the IRS.
Therefore, without considering the precoding, we can simply design the IRS phase shifts to maximize the minimum channel gain for all users:\begin{subequations}\label{eq:initial theta}\begin{align}\label{eq:initial obj}
\underset{\bm{\theta}_0}{\max}~~&\underset{k}{\min}~~~\left\|\mathbf{h}^H_k + \mathbf{h}^H_{\mathrm{r}k}\bm{\Theta}_0\mathbf{G}\right\|^2 \\
&~\mathrm{s.t.}~~~~\bm{\Theta}_0=\mathrm{diag}\{\bm{\theta}_0\}, \left|\theta_n\right|=1, \forall n.
\end{align}
\end{subequations}

For the algorithm development, we rewrite the objective function (\ref{eq:initial obj}) in a more concise format as \begin{equation}
\begin{aligned}
f_k(\bm{\theta}_0) &\triangleq \left\|\mathbf{h}^H_k + \mathbf{h}^H_{\mathrm{r}k}\bm{\Theta}_0\mathbf{G}\right\|^2 = \left\|\mathbf{h}^H_k + \bm{\theta}_0^H\mathbf{G}_k\right\|^2 \\
& = \bm{\theta}_0^H\mathbf{G}_k\mathbf{G}_k^H\bm{\theta}_0 + \bm{\theta}_0^H\mathbf{G}_k\mathbf{h}_k + \mathbf{h}^H_k\mathbf{G}_k^H\bm{\theta}_0 + \mathbf{h}^H_k\mathbf{h}_k,
\end{aligned}
\end{equation}
where $\mathbf{G}_k \triangleq \mathrm{diag}\left\{\mathbf{h}^H_{\mathrm{r}k}\right\}\mathbf{G}, \forall k$, for simplicity.
Then, using the log-sum-exp approximation, problem (\ref{eq:initial theta}) is reformulated as
\begin{subequations}
\label{eq:initial theta 2}
\begin{align}
&\underset{\bm{\theta}_0}{\min}~~~g(\bm{\theta}_0) \triangleq \varepsilon\log\sum_{k=1}^K\exp\left[\frac{-f_k(\bm{\theta}_0)}{\varepsilon}\right]\\
\label{eq:umc}
&~\mathrm{s.t.}~~~\left|\theta_n\right|=1, \forall n,
\end{align}
\end{subequations}
where the unit modulus constraint (\ref{eq:umc}) forms an $N$-dimensional complex circle manifold
\be \mathcal{M} = \left\{\bm{\theta}_0 \in \mathbb{C}^N: \left|\theta_n\right| = 1, \forall n\right\},\ee with the tangent space
\be
T_{\bm{\theta}_0}\mathcal{M} = \left\{\mathbf{p} \in \mathbb{C}^N: \mathfrak{R}\left\{\mathbf{p}\circ \bm{\theta}_0^*\right\} = \mathbf{0}_N\right\}.
\ee
The Riemannian gradient of $g(\bm{\theta}_0)$ is thus given by
\be
\mathrm{grad}_{\bm{\theta}_0}g = \bigtriangledown_{\bm{\theta}_0}g - \mathfrak{R}\left\{\bigtriangledown_{\bm{\theta}_0}g\circ \bm{\theta}_0^*\right\}\circ \bm{\theta}_0,
\ee
where the Euclidean gradient $\bigtriangledown_{\bm{\theta}_0}g$ can be calculated as
\be
\bigtriangledown_{\bm{\theta}_0}g = \frac{\sum_{k=1}^K\left\{\exp\left[\frac{-f_k(\bm{\theta}_0)}{\varepsilon}\right]
\left(-2\mathbf{G}_k\mathbf{G}_k^H\bm{\theta}_0-2\mathbf{G}_k\mathbf{h}_k\right)\right\}}
{\sum_{k=1}^K\exp\left[\frac{-f_k(\bm{\theta}_0)}{\varepsilon}\right]}.
\ee
Then, following the same procedure as in Algorithm 1, the initialization $\bm{\theta}_0$ can be easily obtained.

\subsection{Complexity Analysis}

In this subsection, we provide a brief complexity analysis for the proposed joint symbol-level precoding and reflecting design algorithms.
The complexity to obtain the initial $\bm{\theta}_0$ is at most $\mathcal{O}\{N^{1.5}\}$ using the RCG algorithm.
For the power minimization problem, the complexity to solve for the precoder $\mathbf{x}_m$ by the gradient projection algorithm is $\mathcal{O}\{M^3\}$, and the worst-case computation for the RCG algorithm is of order $\mathcal{O}\{(2N)^{1.5}\}$.
Therefore, the total computational complexity of Algorithm \ref{alg:2} is $\mathcal{O}\{\Omega^KM^3+(2N)^{1.5}\}$.
The symbol-level precoding algorithm for the QoS balancing problem is derived by solving the corresponding power minimization problem, and the reflecting designs for these two problems are similar.
Thus, the complexity to solve the QoS balancing problem is the same as the power minimization problem.

For comparison, the complexity of the joint linear block-level precoding and reflecting design in \cite{Wu TWC 2019} should also be analyzed.
In \cite{Wu TWC 2019}, the linear block-level precoding design is a second-order cone program (SOCP) problem with a complexity of order $\mathcal{O}\left\{M^{4.5}K^{3.5}\right\}$, and the reflecting design is solved using a semidefinite relaxation (SDR) algorithm, whose complexity is of order $\mathcal{O}\left\{K^{3.5}N^{2.5}+K^{2.5}N^{3.5}\right\}$.
Therefore, the total complexity of the proposed algorithm in [10] is of order
$\mathcal{O}\left\{M^{4.5}K^{3.5}+K^{3.5}N^{2.5}+K^{2.5}N^{3.5}\right\}$. We can observe that the algorithm in \cite{Wu TWC 2019} has polynomial complexity in the number of transmit antennas, reflecting elements and users, while the proposed algorithm has exponential complexity in the number of users and the base is the order of the modulation.
However, the significant performance gains obtained by the proposed approach make it worth considering despite the resulting complexity; in cases where the number of users or the modulation order are not very large, the required computational load is still manageable.

\begin{figure}[!t]
  \centering
  \includegraphics[width = 0.45\textwidth]{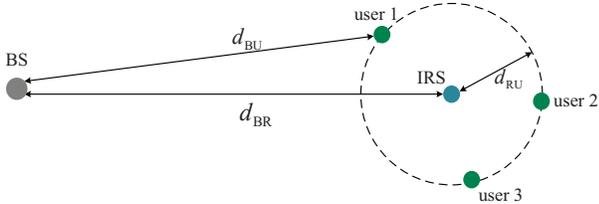}
  \vspace{-0.2 cm}
  \caption{Simulation setup for multiuser case.}\label{fig:set1}
  \vspace{-0.2 cm}
\end{figure}
\begin{figure}[!t]
  \centering
  \includegraphics[width = 0.45\textwidth]{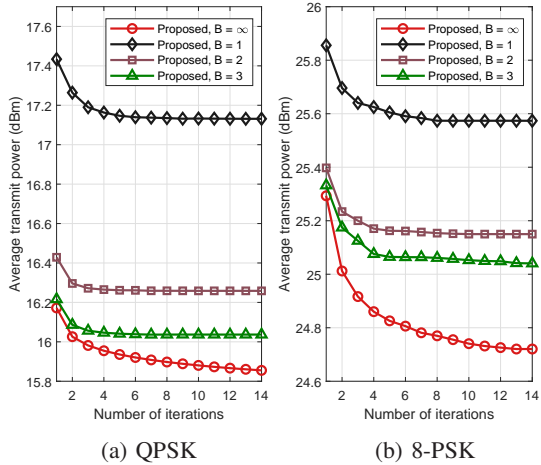}
  \vspace{-0.0 cm}

  \small{(a) QPSK   \;\;\;\;\;\; \;\;\;\;\;\; \;\;\;\;\;\;\;\;\;\;  (b) 8-PSK}
  \caption{Average transmit power versus the number of iterations ($K=3$ users, $N=64$ reflecting elements, $M=6$ transmit antennas, $\Gamma = 10$ dB).}\label{fig:delta_iter}
  \vspace{-0.2 cm}
\end{figure}

\section{Simulation Results}
\label{sec:simulation}
\vspace{0.2 cm}

In this section, we provide extensive simulation results to illustrate the effectiveness of our proposed algorithms.
For simplicity, we assume either QPSK ($\Omega = 4$) or 8-PSK modulation ($\Omega = 8$) are used.
The QoS requirement and the noise power for $K = 3$ users is the same, i.e., $\Gamma = \Gamma_k, \forall k, \sigma^2 = \sigma_k^2 = -80$dBm, $\forall k$.
The transmit antenna array at the BS is assumed to be a uniform linear array with antenna spacing given by $\lambda/2$.
The distance-dependent path loss is modeled as $\mathrm{PL}(d) = C_0\left(\frac{d}{d_0}\right)^{-\alpha}$, where $C_0 = -30$dB is the path loss for the reference distance $d_0 = 1$m, $d$ is the link distance, and $\alpha$ is the path-loss exponent.
In addition, the small-scale Rician fading channel model for all channels is assumed, which consists of line-of-sight (LoS) and non-LoS (NLoS) components.
The channels from the BS to the IRS can be expressed as
\begin{equation}
\mathbf{G} = \sqrt{\frac{\kappa}{\kappa+1}}\mathbf{G}^{\mathrm{LoS}}+
\sqrt{\frac{1}{\kappa+1}}\mathbf{G}^{\mathrm{NLoS}},
\end{equation}
where $\kappa$ is the Rician factor set as 3dB, $\mathbf{G}^{\mathrm{LoS}}$ is the LoS component which depends on the geometric settings, and $\mathbf{G}^{\mathrm{NLoS}}$ is the NLoS Rayleigh fading component.
The MISO channels $\mathbf{h}_k$ and $\mathbf{h}_{\mathrm{r}k}, k = 1,\ldots, K$, are assumed to obey a similar model, consisting of LoS and NLoS components.

The geometry of the following simulations is shown in Fig. \ref{fig:set1}, from a top-down view.
The IRS is typically deployed close to the users to provide them with performance gains, since they may suffer from blockage and severe attenuation.
Therefore, we set the distance between the BS and the IRS as $d_{\mathrm{BR}} = 50$m, the distance between the IRS and the users as $d_{\mathrm{RU}} = 3$m, and the distance between the BS and each user as $d_{\mathrm{BU}}$, which lies in the interval $[d_{\mathrm{BR}}-d_{\mathrm{RU}},d_{\mathrm{BR}}+d_{\mathrm{RU}}]$.
The users are randomly distributed on the dashed circle in Fig. \ref{fig:set1}.
Considering the link distance, the path-loss exponents for $\mathbf{h}_k$, $\mathbf{h}_{\mathrm{r}k}$ and $\mathbf{G}$ are set as 3.5, 2.8, and 2.5, respectively.
Similar settings are widely adopted in existing works, e.g., \cite{Wu TWC 2019}.

\subsection{Power Minimization Problem}

\begin{figure}[!t]
\centering
\subfigure[QPSK]{
\begin{minipage}{0.45\textwidth}
\centering
\includegraphics[width = \textwidth]{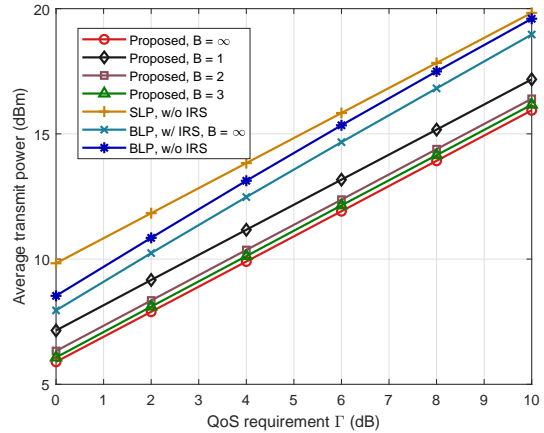}
\vspace{-0.3 cm}
\label{fig:power_SNR4}
\end{minipage}
}
\subfigure[8-PSK]{
\begin{minipage}{0.45\textwidth}
\centering
\includegraphics[width = \textwidth]{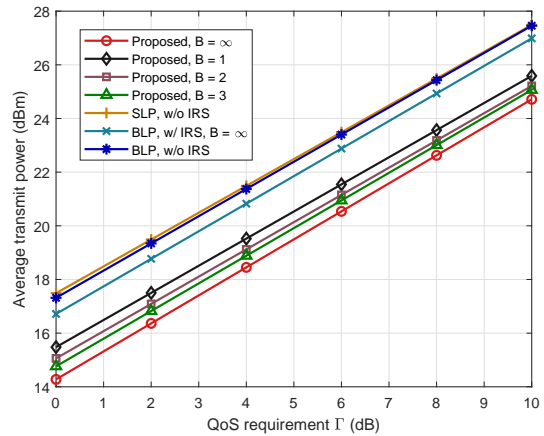}
\vspace{-0.3 cm}
\label{fig:power_SNR8}
\end{minipage}
}
\caption{Average transmit power versus QoS requirement $\Gamma$ ($K=3$ users, $N=64$ reflecting elements, $M=6$ transmit antennas).}
\label{fig:power_SNR}
\vspace{-0.2 cm}
\end{figure}

\begin{figure}
\centering
  \includegraphics[width = 0.45\textwidth]{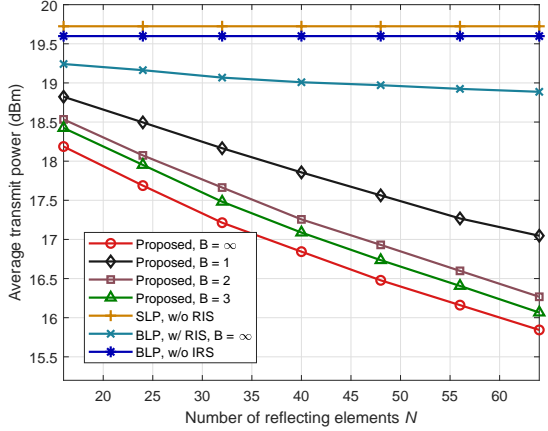}
  \vspace{-0.2 cm}
  \caption{Average transmit power versus the number of reflecting elements $N$ ($K=3$ users, $M=6$ transmit antennas, $\Gamma = 10$ dB).}
  \label{fig:power_N} \vspace{-0.2 cm}
\end{figure}
In this subsection, we illustrate the simulation results for the power minimization problem.
We first show in Fig. \ref{fig:delta_iter} the convergence of our proposed algorithm for the cases where the IRS has continuous, 1-bit, 2-bit, and 3-bit phase shifters, i.e., $B = \infty, 1, 2, 3$, respectively.
It can be observed that our proposed algorithm converges within 14 iterations for all schemes and the low-resolution cases have much faster convergence. We also see that the QPSK case converges more quickly than the 8-PSK case and achieves a lower transmit power, but there is a relatively large gap between the continuous and low-resolution cases.
These convergence results are encouraging for a low-complexity implementation.

\begin{figure}[!t]
  \centering
  \includegraphics[width = 0.45\textwidth]{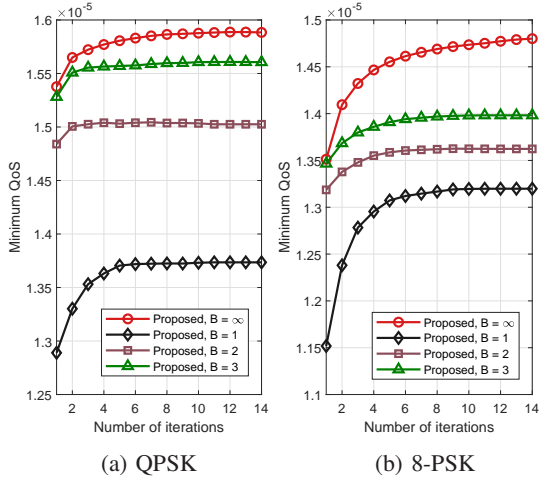}
  \vspace{-0.0 cm}

  \small{(a) QPSK   \;\;\;\;\;\; \;\;\;\;\;\; \;\;\;\;\;\;\;\;\;\;  (b) 8-PSK}
  \caption{Minimum QoS versus the number of iterations ($K=3$ users, $N=64$ reflecting elements, $M=6$ transmit antennas, $P = 20$ dBm).}\label{fig:cost_iter}
  \vspace{-0.2 cm}
\end{figure}

In Fig. \ref{fig:power_SNR}, we show the average transmit power versus the QoS requirement $\Gamma$.
In order to demonstrate the effectiveness of our proposed joint symbol-level precoding and reflecting design, we also include: \textit{i}) Symbol-level precoding without the aid of the IRS (denoted as ``SLP, w/o IRS''); \textit{ii}) linear block-level precoding with the aid of IRS and continuous phase shifters \cite{Wu TWC 2019} (denoted as ``BLP, w/ IRS, $B = \infty$''); \textit{iii}) linear block-level precoding without the aid of the IRS \cite{Wiesel TSP 2006} (denoted as ``BLP, w/o IRS'').
It can be seen from Fig. \ref{fig:power_SNR} that our proposed scheme requires less transmit power than the ``SLP, w/o IRS'' approach in both the QPSK and 8-PSK cases, which validates the effectiveness of using IRS in the symbol-level precoding systems.
We can also observe that the proposed joint symbol-level precoding and reflecting algorithm outperforms the ``BLP, w/ IRS, $B = \infty$'' and ``BLP, w/o IRS'' approaches, which verifies the performance improvement due to symbol-level precoding.
In addition, it is noted that with increasing $B$, better system performance can be achieved.
Moreover, the 3-bit quantized solution can provide performance similar to the ideal unquantized solution, thus providing a favorable trade-off between cost and performance.
Beyond $B=3$ bits, the extra cost and complexity associated with using higher-resolution IRS are not warranted given the very marginal increase in system performance.
On the other hand, the performance gap between the different approaches in Fig. \ref{fig:power_SNR8} is relatively smaller than that in Fig. \ref{fig:power_SNR4}, which indicates that higher-order modulation can compensate for the performance loss due to low-resolution phase shifts or non-optimal precoders.
However, significantly more transmit power is required to support higher-order modulation.

\begin{figure}[!t]
\centering
\subfigure[QPSK]{
\begin{minipage}{0.45\textwidth}
\centering
\includegraphics[width = \textwidth]{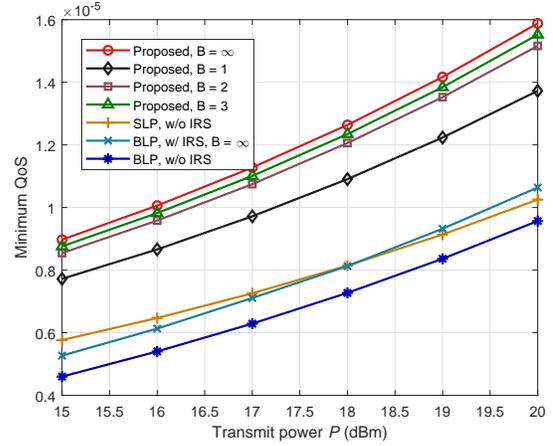}
\vspace{-0.2 cm}
\label{fig:cost_P4}
\end{minipage}
}
\subfigure[8-PSK]{
\begin{minipage}{0.45\textwidth}
\centering
\includegraphics[width = \textwidth]{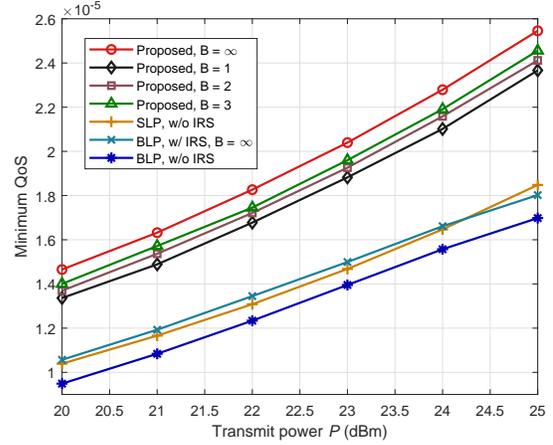}
\vspace{-0.2 cm}
\label{fig:cost_P8}
\end{minipage}
}
\caption{Minimum QoS versus the transmit power $P$ ($K=3$ users, $N=64$ reflecting elements, $M=6$ transmit antennas).}
\label{fig:cost_P}
\vspace{-0.2 cm}
\end{figure}

Next, we focus on QPSK modulation and present the average transmit power versus the number of reflecting elements $N$ in Fig. \ref{fig:power_N}.
The same relationship can be observed as in Fig. \ref{fig:power_SNR}.
We observe that as the number of reflecting elements increases, the average transmit power is greatly reduced, and the reduction is more pronounced for our proposed SLP algorithms compared with linear block-level precoding.
This supports the main idea of our paper, that the combination of SLP and IRS provides symbiotic benefits.

\subsection{QoS Balancing Problem}

In this subsection, we present simulations for the QoS balancing problem.
The convergence performance is similar to that observed for the power minimization problem in Fig. \ref{fig:cost_iter}.
It is seen that all schemes converges within 14 iterations, which indicates favorable computational complexity.

\begin{figure}[]
\centering
\subfigure[QPSK]{
\begin{minipage}{0.45\textwidth}
\centering
\includegraphics[width = \textwidth]{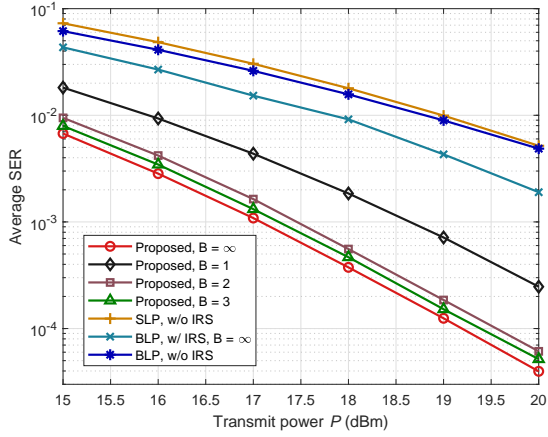}
\vspace{-0.3 cm}
\label{fig:SER_P4}
\end{minipage}
}
\subfigure[8-PSK]{
\begin{minipage}{0.45\textwidth}
\centering
\includegraphics[width = \textwidth]{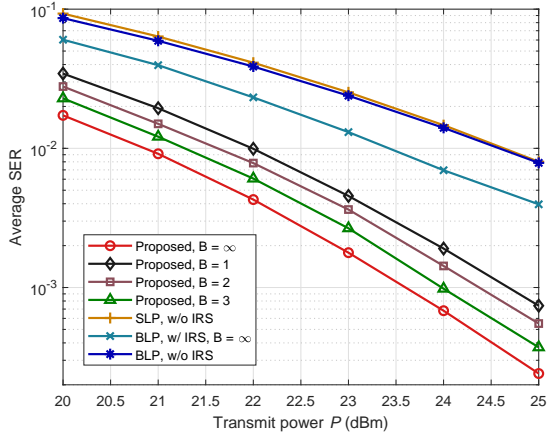}
\vspace{-0.3 cm}
\label{fig:SER_P8}
\end{minipage}
}
\caption{Average SER versus the transmit power $P$ ($K=3$ users, $N=64$ reflecting elements, $M=6$ transmit antennas).}
\label{fig:SER_P}
\vspace{-0.6 cm}
\end{figure}

\begin{figure}[]
  \centering
  \includegraphics[width = 0.45\textwidth]{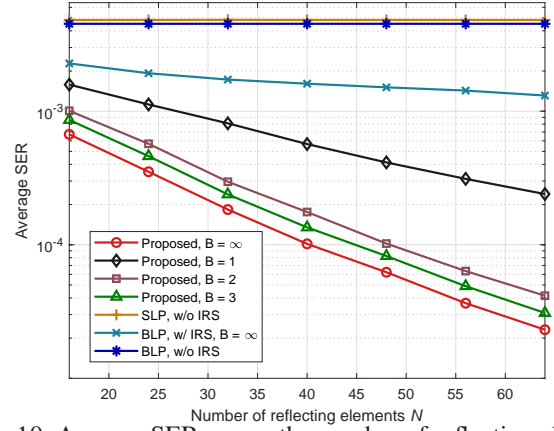}
  \vspace{-0.3 cm}
  \caption{Average SER versus the number of reflecting elements $N$ ($K=3$ users, $N=64$ reflecting elements, $M=6$ transmit antennas, $P = 20$ dBm).}
   \label{fig:SER_N}
  \vspace{-0.2 cm}
\end{figure}

In Fig. \ref{fig:cost_P}, we plot the worst-case performance $t$ achieved for the QoS balancing problem versus the transmit power $P$.
As the transmit power increases, the minimum QoS of all methods increases, which means that the distance between the received noise-free signal and its decision boundaries becomes larger.
The minimum QoS achieved by our proposed joint symbol-level precoding and reflecting algorithm is dramatically larger than the other competitors, which further supports the benefit of using IRS together with symbol-level precoding.

In order to demonstrate the QoS improvement in a more intuitive and natural way with a familiar metric, in Fig. \ref{fig:SER_P} we present the average SER versus the transmit power.
Obviously, the larger QoS requirement (i.e., $\Gamma$), which results in a larger distance between the received signal and its decision boundaries, leads to better performance in terms of a lower SER.
This relationship can be verified by comparing Figs. \ref{fig:cost_P} and \ref{fig:SER_P}.
More importantly, the improvement in the SER performance of our proposed algorithm is also very remarkable.
When the 3-bit IRS can offer close to $10^{-4}$ SER, the symbol-level precoding system without IRS provides only $10^{-2}$ SER.
Therefore, utilizing the QoS requirement $\Gamma$ as the performance metric for optimizing the IRS-enhanced symbol-level precoding systems is reasonable and effective.

Next, we show the average SER versus the number of reflecting elements $N$ for QPSK modulation in Fig. \ref{fig:SER_N}.
Since the larger IRS can provide larger beamforming/reflecting gains, we observe that the average SER decreases for all schemes with increasing $N$.
Moreover, our proposed schemes always achieve significantly better SER performance for different IRS sizes.

\subsection{Impact of IRS Location}

\begin{figure}[!t]
  \centering
  \includegraphics[width = 0.45\textwidth]{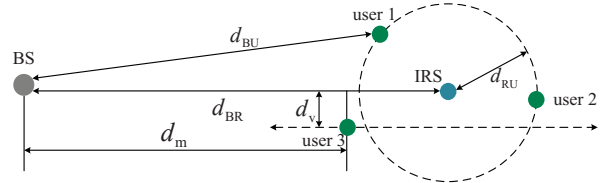}
  \vspace{-0.3 cm}
  \caption{Simulation setup for mobile user case.}\label{fig:set2}
  \vspace{-0.2 cm}
\end{figure}

\begin{figure}[!t]
  \centering
  \includegraphics[width = 0.45\textwidth]{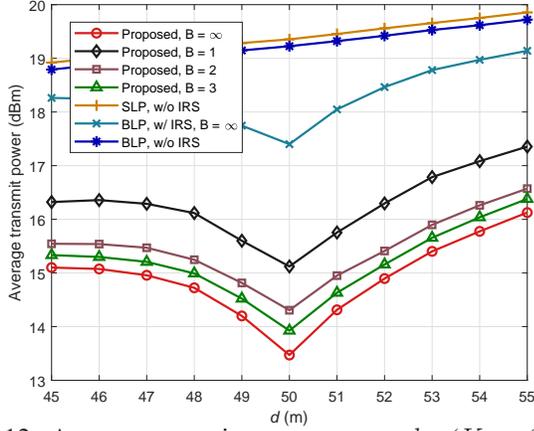}
  \vspace{-0.3 cm}
  \caption{Average transmit power versus $d_\mathrm{m}$ ($K=3$ users, $N=64$ reflecting elements, $M=6$ transmit antennas, SNR $ = 10$ dB).}
   \label{fig:p_d}
  \vspace{-0.2 cm}
\end{figure}

In order to demonstrate the impact of IRS location, we focus on QPSK modulation and simulate a case where the position of one of the users changes along a horizontal line parallel to the line between the BS and IRS.
As shown in Fig. \ref{fig:set2}, user 3 moves along the dashed line and the vertical distance between it and the BS-IRS link is $d_\mathrm{v} = 0.5$m.
Let $d_\mathrm{m}$ be the horizontal distance between the BS and user 3.
The other two users are still located 3m from the IRS.
In Figs. \ref{fig:p_d} and \ref{fig:cost_d}, we show the system performance as a function of $d_\mathrm{m}$.
We observe that the proposed schemes always outperform other benchmarks, and the best performance is achieved when the user moves closest to the IRS, i.e., $d_\mathrm{m} = 50$m, since a larger reflection gain is obtained when the user is closest to the IRS.
Moreover, when the users move closer together, MUI may become stronger, which can be effectively exploited by symbol-level precoding.

\begin{figure}[!t]
  \centering
  \includegraphics[width = 0.45\textwidth]{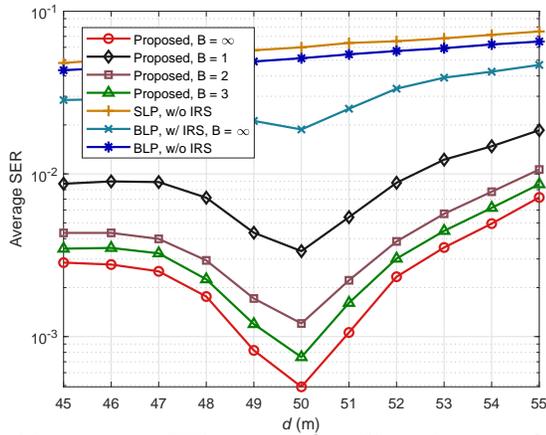}
  \vspace{-0.3 cm}
  \caption{Average SER versus $d_\mathrm{m}$ ($K=3$ users, $N=64$ reflecting elements, $M=6$ transmit antennas, $ P= 15$ dBm).}
   \label{fig:cost_d}
  \vspace{-0.2 cm}
\end{figure}

\section{Conclusions}
\label{sec:conclusion}

In this paper, we investigated IRS-enhanced wireless networks, where an IRS is deployed to assist the multi-user MISO communication system, which employs symbol-level precoding to exploit the multi-user interference.
In particular, we considered the joint symbol-level precoding and reflecting design problems for IRS-enhanced MU-MISO systems.
Efficient iterative algorithms were proposed to solve the power minimization and QoS balancing problems.
The gradient-projection-based and Riemannian conjugate gradient (RCG)-based algorithms were used to design the symbol-level precoding and IRS phase shifts, respectively.
The simulation results illustrated that our proposed algorithms exhibit remarkably better performance in terms of power-savings and SER-reductions.
These positive results have confirmed that the employment of IRS in symbol-level precoding systems can provide more efficient multi-user interference exploitation by intelligently manipulating the multi-user channels.

\begin{appendices}
\section{}
\begin{IEEEproof}[Proof of Proposition 1]
We assume the optimal solution of the QoS balancing problem (\ref{eq:SINR balancing precoding}) is $\mathbf{x}_m^*, \forall m$, and $t^*$, and we denote the transmit power allocated to the $m$-th precoder as $P_m \triangleq \left\|\mathbf{x}_m^*\right\|^2$.
If the power allocations $P_m, \forall m$, are known, the QoS balancing problem can be divided into $\Omega^K$ sub-problems and the $m$-th sub-problem is written as
\begin{subequations}
\label{eq:proof 1}
\begin{align}
  &\underset{\mathbf{x}_m,t_m}{\max}~~~t_m \\
  &\mathrm{s.t.}~~~~~\mathfrak{R}\{\widetilde{\mathbf{h}}^H_k
  \mathbf{x}_me^{-j\angle{s_{m,k}}}\}\tan \Phi \\ &~~~~~~~~~~~~~~~~~~~~~~~~~-\left|\mathfrak{I}\{\widetilde{\mathbf{h}}^H_k\mathbf{x}_me^{-j\angle{s_{m,k}}}\}\right| \geq t_m, \forall k,  \non\\
  &~~~~~~~~\left\|\mathbf{x}_m\right\|^2 \leq P_m. \label{eq:proof b}
\end{align}
\end{subequations}
It can be easily verified that the inequality constraint (\ref{eq:proof b}) holds with equality at the optimal $\mathbf{x}_m^{*}$. Moreover, $\mathbf{x}_m^{*}$ is also the optimal solution for the following power minimization problem:
\begin{subequations}
\label{eq:proof 2}
\begin{align}
  &\underset{\mathbf{x}_m}{\min}~~~\| \mathbf{x}_m \|^2 \label{eq:proof 2a} \\
  &\mathrm{s.t.}~~~~\mathfrak{R}\{\widetilde{\mathbf{h}}^H_k
  \mathbf{x}_me^{-j\angle{s_{m,k}}}\}\tan \Phi\label{eq:proof 2b}\\
  &~~~~~~~~~~~~~~~~~~~~~~~~~-\left|\mathfrak{I}\{\widetilde{\mathbf{h}}^H_k\mathbf{x}_me^{-j\angle{s_{m,k}}}\}\right| \geq t_m^*, \forall k.\non
\end{align}
\end{subequations}
To prove this statement by contradiction, we start by assuming that $\mathbf{x}_m^{*}$ is not optimal for (\ref{eq:proof 2}) and there exists another $\overline{\mathbf{x}}_m$ satisfying (\ref{eq:proof 2b}) and requiring less power, i.e., $\| \overline{\mathbf{x}}_m\| < P_m$. Then, we can scale up $\overline{\mathbf{x}}_m$ to let the power constraint (\ref{eq:proof b}) become equal and provide a higher $t_m$ in (\ref{eq:proof 1}), in which case $\mathbf{x}_m^{*}$ is not optimal any more, and this results in a contradiction.

If we use an arbitrary $t_0 > 0 $ as the QoS requirement in (\ref{eq:proof 2b}) instead of $t_m^*$, the optimal solution $\mathbf{x}_m^{\star}$ of this power minimization problem is the scaled version of $\mathbf{x}_m^{*}$, since the constraint (\ref{eq:proof 2b}) is a linear function. More specifically, $\mathbf{x}_m^* = \frac{\sqrt{P_m}\mathbf{x}_m^\star}{\left\|\mathbf{x}_m^\star\right\|}$ with the optimal power $P_m$.

If we set the QoS requirement as $t_0 = \sigma_k\sqrt{\Gamma_k}\tan\Phi$, then (\ref{eq:proof 2}) has the same format as (\ref{eq:xm}), which implies that the optimal $\mathbf{x}_m^\star$ for the power minimization problem (\ref{eq:xm}) is also a scaled version of the optimal $\mathbf{x}_m^*$ for the QoS balancing problem (\ref{eq:SINR balancing precoding}).
Proposition 1 is therefore proved.

\end{IEEEproof}
\end{appendices}

\end{document}